\documentclass[aip, amsmath, amssymb, reprint, floatfix, nofootinbib]{revtex4-1}

\usepackage{graphicx}
\usepackage{dcolumn}
\usepackage{mathptmx}
\usepackage[utf8]{inputenc}
\usepackage[T1]{fontenc}
\usepackage{algorithm,algpseudocode}
\usepackage{bbm}
\usepackage[english]{babel}
\usepackage{amsthm}
\usepackage{mathtools}
\usepackage{subcaption}
\newtheorem{theorem}{Theorem}

\renewcommand{\vec}[1]{\mathbf{#1}}
\renewcommand{\u}{\mathbf{u}}
\renewcommand{\v}{\mathbf{w}}

\newcommand{\game}[1]{\texttt{#1}}
\newcommand{\eqpt}{EqPt}
\newcommand{\ang}[2]{{\cos}{\angle}({#1},{#2})}
\newcommand{\inprod}[2]{{\langle}{#1},{#2}{\rangle}}

\DeclarePairedDelimiter{\abs}{\big\lvert}{\big\rvert}
\DeclarePairedDelimiter{\norm}{\lVert}{\rVert}

\renewcommand{\lambda}{R}

\begin{document}

\preprint{AIP/123-QED}

\title{Geometrical Regret Matching}

\author{Sizhong Lan}
\affiliation{$^1$China Mobile Research Institute, Beijing 100053, China}
\email{lsz@nzqrc.cn}

\date{\today}

\begin{abstract}
    We argue that the existing regret matchings for Nash equilibrium approximation conduct ``jumpy'' strategy updating when the probabilities of future plays are set to be proportional to positive regret measures.
    We propose a geometrical regret matching which features ``smooth'' strategy updating.
    Our approach is simple, intuitive and natural.
    The analytical and numerical results show that, continuously and ``smoothly'' suppressing ``unprofitable'' pure strategies is sufficient for the game to evolve towards Nash equilibrium, suggesting that in reality the tendency for equilibrium could be pervasive and irresistible.
    Technically, iterative regret matching gives rise to a sequence of adjusted mixed strategies for our study its approximation to the true equilibrium point.
    The sequence can be studied in metric space and visualized nicely as a clear path towards an equilibrium point.
    Our theory has limitations in optimizing the approximation accuracy.
\end{abstract}

\maketitle

\section{Introduction}
In 2000 Hart and Mas-Colell proposed an iterative algorithm called regret matching to approximate a correlated equilibrium~\cite{Hart2000,Hart2001}.
The players keep tracking the regrets on the past plays and making the future plays with probabilities proportional to positive regret measures.
This algorithm is particularly natural in that the players don't have to ply about their opponents' payoff functions, as opposed to the non-adaptive variety, e.g. the celebrated Lemke-Howson algorithm~\cite{Lemke1964,Lemke1965} which takes two players' payoff matrices as input and pinpoints equilibrium points as output.
In other words, regret matching allows the players to reach equilibrium without strict mediation from beyond themselves.
The concept of regret matching has since been adopted and developed by other algorithms~\cite{cfr2008}.

Iterative regret matching can be seen as continuing updating of mixed strategy with regret information:
the mixed strategy to update is statistical structure of the whole past plays, and the mixed strategy updated will determine the probabilities of plays in the immediate future.
Regret matching is essentially a function between them.
In general the existing regret-matching functions update the mixed strategy proportional to positive regret measures, meaning that each matching is a ``strategy jump'' and the past mixed strategy has little relevance except for it being used for regret evaluation.
In this paper, we try to propose a ``smoother'' regret matching with more respect to the past plays.
Partly inspired by the mapping $T$ in Nash's existential proof of equilibrium point~\cite{Nash1951}, our function for regret matching takes the form of

\begin{equation*}
    \vec{s}'_i=\frac{\vec{s}_i+r_i\lambda_{i}(\vec{s}_i)}{1+r_i\abs{\lambda_{i}(\vec{s}_i)}}\text{, where }r_i>0.
\end{equation*}

As shown in FIG~\ref{fig:angle_close_up} our regret-matching function can be interpreted geometrically if we treat mixed strategies and regrets on pure strategies both as vectors:
each regret matching will ``push'' mixed strategy vector towards regret vector by a small angle, and the smaller $r_i$ is the ``smoother'' matching will be.
Our matching leads to clear analytical results for regret minimization.
The test results in FIG~\ref{fig:3X3-1eq2sp},~\ref{fig:60X40-conv} and~\ref{fig:nP-path} show that, our matching delivers strong numerical convergence to Nash equilibrium both for two-person and many-person non-cooperative games.
And as shown in FIG~\ref{fig:3X3-1eq2sp-attractor}, the convergence is independent of the initial mixed strategies to start game with.

\begin{figure}[h]
\centering
\includegraphics[width=\linewidth]{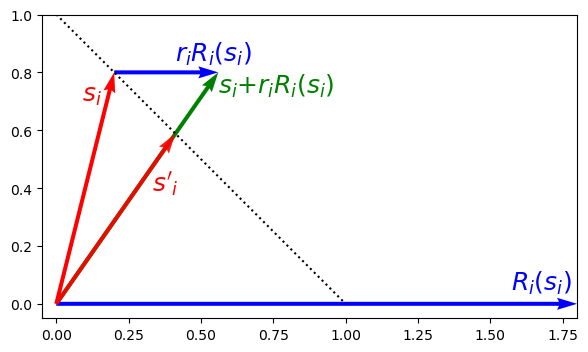}
\caption{\label{fig:angle_close_up}
    An illustrative example for the geometrical relation of mixed strategy $\vec{s}_i$, regret vector $\lambda_i(\vec{s}_i)$ and regret-matched $\vec{s}'_i$.
    The dotted line represents the simplex in $\mathbb{R}^2$.
    In this example regret vector $\lambda_i(\vec{s}_i)$ must be to the direction of one vertex since the other one is ``less profitable than average'' such that the corresponding regret component in vector $\lambda_i(\vec{s}_i)$ is zero.
}
\end{figure}

\begin{figure}[h]
\centering
\includegraphics[width=\linewidth]{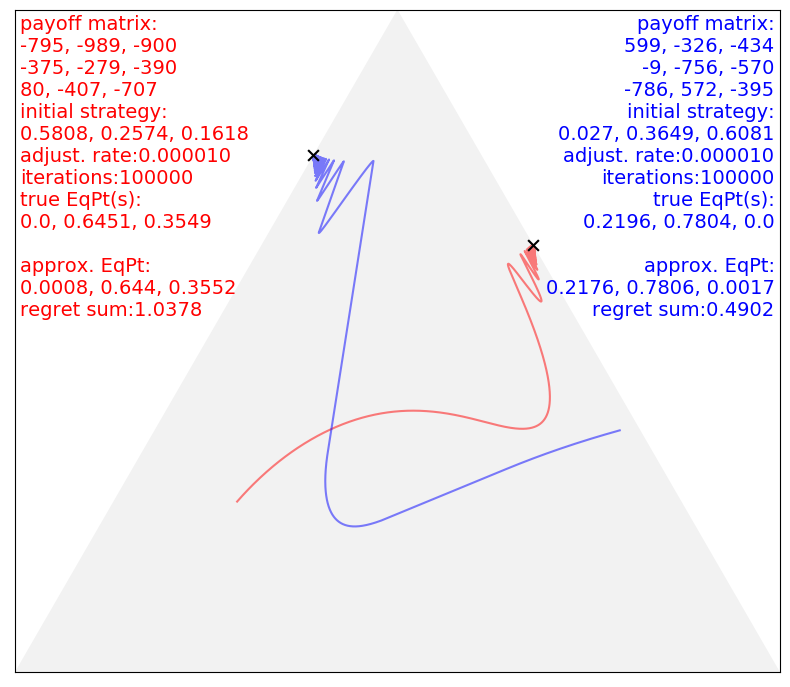}
\caption{\label{fig:3X3-1eq2sp}
The pair of ``strategy paths'' on the simplex of $\mathbb{R}^3$ for a two-person game.
The game has the code name \game{3X3-1eq2sp} to refer to, indicating that the game has one single equilibrium point and either player can use three pure strategies.
}
\end{figure}

\begin{figure}[h]
\centering
\includegraphics[width=\linewidth]{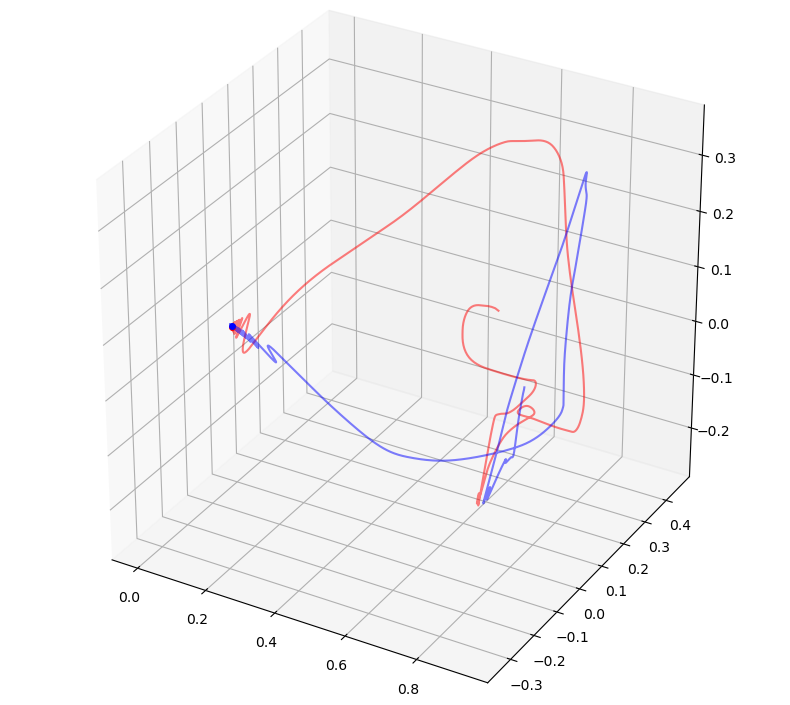}
\caption{\label{fig:60X40-conv}
The pair of ``strategy paths'' for a two-person game where two players use $60$ and $40$ pure strategies respectively.
The $60$ or $40$ dimensions of ``strategy path'' data are reduced with PCA to three dimensions for the visualization in $\mathbb{R}^3$ .
}
\end{figure}

\begin{figure}[h]
\centering
\includegraphics[width=\linewidth]{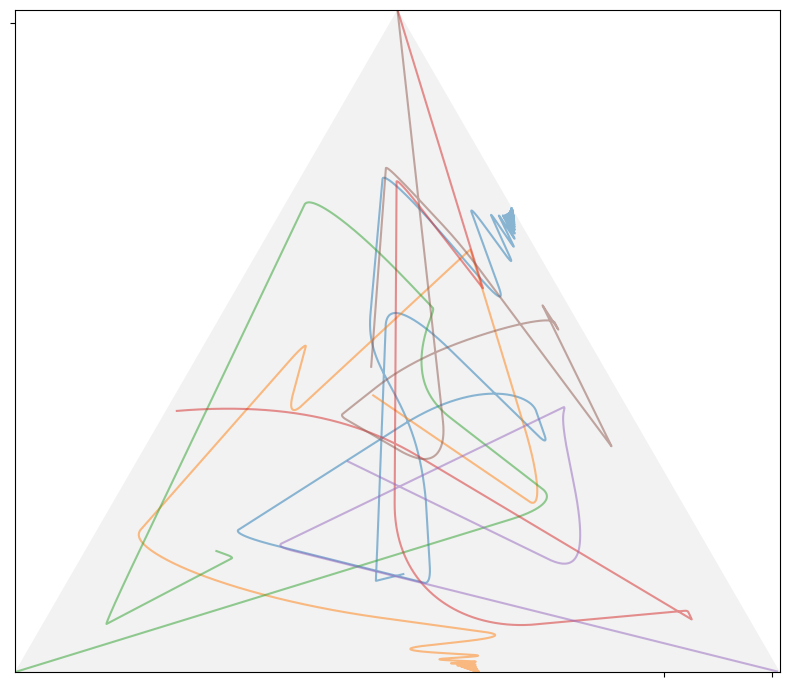}
\caption{\label{fig:nP-path}
    The ``strategy paths'' for a six-person game, where each player can use three pure strategies.
}
\end{figure}

\begin{figure}[h]
\centering
\includegraphics[width=\linewidth]{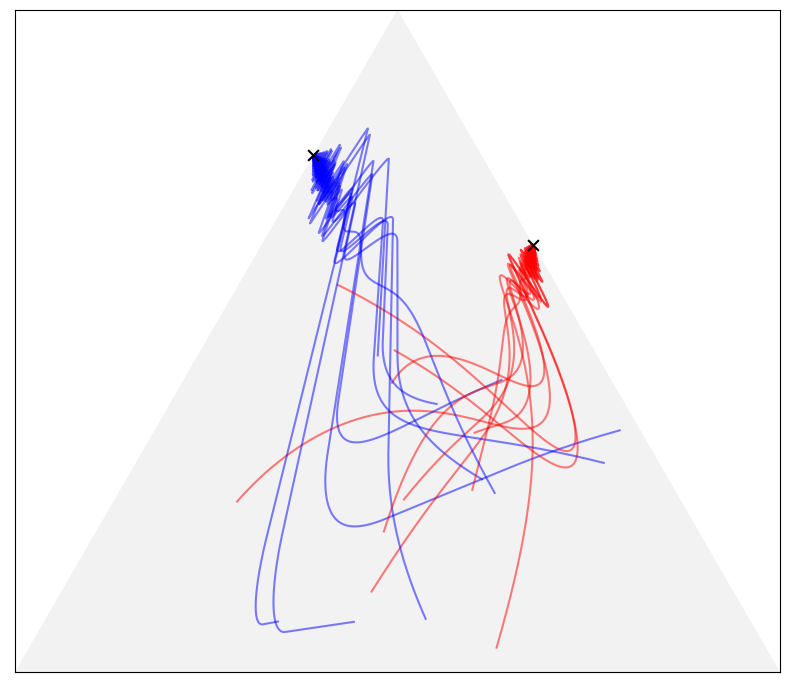}
\caption{\label{fig:3X3-1eq2sp-attractor}
The equilibrium point of game \game{3X3-1eq2sp} is an ``attractor'' for all possible ``strategy paths''.  }
\end{figure}

To the other extreme, we notice that $\vec{s}'_i$ tends to $\lambda_{i}(\vec{s}_i){/}\abs{\lambda_{i}(\vec{s}_i)}$ as $r_i$ tends to positive infinity, reminiscent of the classic matching proportional to positive regret.
In fact, as an important result, our matching will reveal the true rationality underlying the somewhat heuristic "proportional to positive regret": to immediately suppress ``less profitable than average'' pure strategies.
This result agrees well with the observed behavioral tendency, which, together with the smoothness of strategy updating, makes our matching a natural one.

The remainder of paper goes as follows.
In Section~\ref{section:game}, we vectorize the key concepts of non-cooperative game as well as regret measures in preparation of Section~\ref{section:fpi_eqpt}, where our regret matching is elaborated and applied directly upon mixed strategies so as to approximate a Nash equilibrium point.
In Section~\ref{section:2p_game}, the approximation algorithm is verified with the examples both of two-person game and general $n$-person game, and their ``strategy paths" are visualized as shown in FIG~\ref{fig:3X3-1eq2sp}{--}\ref{fig:3X3-1eq2sp-attractor}.
In Section~\ref{section:accuracy}, we address the accuracy issues of equilibrium approximation.
In Section~\ref{section:lambda_variants}, we consider some generalizations for our theory.

\section{Non-cooperative game and regret\label{section:game}}

In this section, we will define the condition of Nash equilibrium as well as the regret measures in the form of vector and its $L1$-norm.

First of all, let us define the familiar concepts of non-cooperative game and set up their notations.

\begin{itemize}
\item There are $n$ players, and each player $i$ has $g_i$ pure strategies. The $i$\textsuperscript{th} player's $j$\textsuperscript{th} pure strategy is denoted by $\pi_{ij}$.

\item The mixed strategy $\vec{s}_i$ of each player $i$ is a point of a probability simplex $\mathbb{I}_i$, which is a compact and convex subset of the real vector space $\mathbb{R}^{g_i}$ of $g_i$ dimensions. That is, $\vec{s}_i{\in}\mathbb{I}_i{\subset}\mathbb{R}^{g_i}$.

\item Each player $i$ has a payoff function $\bar{p}_i:\mathbb{S}{\rightarrow}\mathbb{R}$. Any $\vec{S}{\in}\mathbb{S}$ is an $n$-tuple of mixed strategies, e.g. $(\vec{s}_1,\ldots,\vec{s}_i,\ldots,\vec{s}_n)$ in which each item $\vec{s}_i$ is associated with the $i$\textsuperscript{th} player such that $\vec{s}_i{\in}\mathbb{I}_i$. Therefore, $\mathbb{S}$ is the product of all $n$ players' simplices, i.e. $\mathbb{I}_1{\times}\mathbb{I}_2{\times}{\cdots}{\times}\mathbb{I}_n$.

\item Any player $i$ can, from the given $\vec{S}{\in}\mathbb{S}$, form a new $n$-tuple of mixed strategies, denoted by $\vec{s}_i;\vec{S}{\in}\mathbb{S}$, by unilaterally substituting the $i$\textsuperscript{th} item of $\vec{S}$ with its new mixed strategy $\vec{s}_i$, or as a special case, $\pi_{ij};\vec{S}{\in}\mathbb{S}$ by unilaterally substituting with its $j$\textsuperscript{th} pure strategy. Obviously, $\vec{s}_i;\vec{S}{=}\vec{S}$ if $\vec{S}$ already uses $\vec{s}_i$.

\end{itemize}

Intendedly, these vector-centric definitions and notations allow us to, for any player $i{=}1,2,{\ldots},n$, translate the linearity of its payoff $\bar{p}_i(\vec{S})$ in $\vec{s}_i$ into $\bar{p}_i(\vec{S}){=}\inprod{\vec{s}_i}{\vec{v}_i}$.
Here $\inprod{\cdot}{\cdot}$ is the inner product of vectors, $\vec{s}_i$ is the mixed strategy player $i$ uses in $\vec{S}$ and $\vec{v}_i$ is the player's payoff vector for the $g_i$ vertices of its simplex $\mathbb{I}_i$. Given any $n$-tuple $\vec{S}{\in}\mathbb{S}$, for each player $i$ we have

\begin{equation}
\vec{v}_{i,\vec{S}}=\Big(\bar{p}_i(\pi_{i1};\vec{S}),\ldots,\bar{p}_i(\pi_{ij};\vec{S}),\ldots,\bar{p}_i(\pi_{ig_i};\vec{S})\Big).
\label{eq:vertex_payoff}
\end{equation}

Given any $\vec{S}$, each player can substitute its item in $\vec{S}$ with $\vec{s}_i{\in}\mathbb{I}_i$ to form a new one $\vec{s}_i;\vec{S}{\in}\mathbb{S}$, which means that we can designate for each player $i$ a variant payoff function $p_{i,\vec{S}}$ of substitute strategy $\vec{s}_i$:

\begin{equation}
p_{i,\vec{S}}(\vec{s}_i)=\bar{p}_i(\vec{s}_i;\vec{S}){=}\inprod{\vec{s}_i}{\vec{v}_{i,\vec{S}}}.
\label{eq:payoff_inprod}
\end{equation}

Note that $p_{i,\vec{S}}(\vec{s}_i)$ is the expected value of payoff or simply the average payoff, since $\vec{s}_i$ can be conveniently treated as probability distribution.
From Eq. (\ref{eq:vertex_payoff}) and (\ref{eq:payoff_inprod}) we can derive for each player $i$ a function

\begin{equation}
\lambda_{i,\vec{S}}(\vec{s}_i)=\Big(\varphi_{i1,\vec{S}}(\vec{s}_i),\ldots,\varphi_{ij,\vec{S}}(\vec{s}_i),\ldots,\varphi_{ig_i,\vec{S}}(\vec{s}_i)\Big).
\label{eq:lambda}
\end{equation}

Here each component of vector $\lambda_{i,\vec{S}}(\vec{s}_i)$ is defined by the very function proposed by Nash in his existential proof:

\begin{equation}
\varphi_{ij,\vec{S}}(\vec{s}_i)={\max}\Big\{0,\bar{p}_i(\pi_{ij};\vec{S}){-}p_{i,\vec{S}}(\vec{s}_i)\Big\}.
\label{eq:varphi}
\end{equation}

Each component $\varphi_{ij,\vec{S}}(\vec{s}_i)$ represents the payoff gain for player $i$ when its strategy moves from mixed $\vec{s}_i$ to its $j$\textsuperscript{th} pure strategy.
We say a pure strategy is ``more profitable'' than $\vec{s}_i$ if its $\varphi_{ij,\vec{S}}(\vec{s}_i){>}0$ or otherwise ``less (or equally) profitable''.
Vector $\lambda_{i,\vec{S}}(\vec{s}_i)$ has at least one zero component, since there must be one ``least profitable'' pure strategy as implied by the linearity of $p_{i,\vec{S}}(\vec{s}_i)$ in $\vec{s}_i$.

Given $\vec{S}$, vector $\lambda_{i,\vec{S}}(\vec{s}_i)$ measures, in terms of payoff, how far a substitute $\vec{s}_i$ is from the optimal one.
Specifically, the $L1$-norm of $\lambda_{i,\vec{S}}(\vec{s}_i)$, i.e. $\abs{\lambda_{i,\vec{S}}(\vec{s}_i)}{=}\sum_j^{g_i}\varphi_{ij,\vec{S}}(\vec{s}_i)$, decreases as $p_{i,\vec{S}}(\vec{s}_i)$ increases and would turn zero when $p_{i,\vec{S}}(\vec{s}_i)$ is maximized with the optimal substitute.
Recall that an equilibrium point is an $n$-tuple $\vec{S}$ in which each player's mixed strategy is optimal against those of its opponents\cite{Nash1951}.
That is, an $n$-tuple $\vec{S}^+{\in}\mathbb{S}$ is an equilibrium point, if and only if $\lambda_{i,\vec{S}^+}(\vec{s}^+_i){=}\vec{0}$ or equivalently $\abs{\lambda_{i,\vec{S}^+}(\vec{s}^+_i)}{=}0$ for each item $\vec{s}^+_i$ in $\vec{S}^+$.

Vector $\lambda_{i,\vec{S}}(\vec{s}_i)$ will be the focal point of our study.
We call $\lambda_{i,\vec{S}}(\vec{s}_i)$ the \emph{regret vector} of player $i$ in that the vector indicates how much payoff it would have gained if its strategy moved to the vertices of its simplex $\mathbb{I}_i$, and correspondingly $\abs{\lambda_{i,\vec{S}}(\vec{s}_i)}$ the \emph{regret sum} of player $i$.
We shall occasionally abbreviate equilibrium point by \eqpt{}, and write $p_i(\vec{s}_i)$, $\varphi_{ij}(\vec{s}_i)$ and $\lambda_i(\vec{s}_i)$ as implicitly indexed by given $\vec{S}$.

\section{Regret Matching and Fixed point iteration\label{section:fpi_eqpt}}

In this section, we will conduct regret matching, utilizing regret vector, iteratively on the initial mixed strategies to approximate a Nash equilibrium.
Ideally, this iterative regret matching could develop into a fixed point iteration~\cite{fpi}.

We define each player's regret matching to be a function $\lambda_{i,\vec{S}}:\mathbb{I}_i{\rightarrow}\mathbb{I}_i$:


\begin{equation}
\psi_{i,\vec{S}}(\vec{s}_i)=\frac{\vec{s}_i+r_i\lambda_{i,\vec{S}}(\vec{s}_i)}{1+r_i\abs{\lambda_{i,\vec{S}}(\vec{s}_i)}}.\label{eq:psi}
\end{equation}

Here real number $r_i{>}0$.
By $\vec{s}'_i{=}\psi_{i,\vec{S}}(\vec{s}_i)$ we translate that, on the game with given $n$-tuple $\vec{S}$, if the player $i$ unilaterally changes its mixed strategy to any other $\vec{s}_i$, it will regret that $\vec{s}'_i$ should have been used instead.
We write $\psi_i(\vec{s}_i)$ occasionally.
Our design of $\psi_i$ for regret matching has three important aspects to consider.

\textbf{First, function $\psi_i$ has a geometrical interpretation.}


We can see in FIG~\ref{fig:angle_close_up} that vector $\vec{s}'_i$ is ``closer'' to vector $\lambda_i(\vec{s}_i)$ than vector $\vec{s}_i$ is.
Generally, the ``closeness'' of two vectors, say $\u$ and $\v$, can be measured by their angle $\angle(\u,\v)$, which further can be measured by $\ang{\u}{\v}{=}\frac{\inprod{\u}{\v}}{\norm{\u}\norm{\v}}$.
By Theorem~\ref{theo:vectorangle}, we have $\ang{\vec{s}_i{+}r_i\lambda_i(\vec{s}_i)}{\lambda_i(\vec{s}_i)}{\ge}\ang{\vec{s}_i}{\lambda_i(\vec{s}_i)}$ and thus

\begin{equation}
\ang{\vec{s}'_i}{\lambda_i(\vec{s}_i)}\ge{}\ang{\vec{s}_i}{\lambda_i(\vec{s}_i)}.
\end{equation}

\begin{theorem}
For any non-zero real vectors $\u$, $\v$ and $\u{+}r\v$ where real $r{>}0$, $\ang{\u{+}r\v}{\v}{\ge}\ang{\u}{\v}$.\label{theo:vectorangle}
\end{theorem}
\begin{proof}
    By definition, we have $\ang{\u{+}r\v}{\v}{=}\frac{\inprod{\u{+}r\v}{\v}}{\norm{\u{+}r\v}\norm{\v}}$ and $\ang{\u}{\v}{=}\frac{\inprod{\u}{\v}}{\norm{\u}\norm{\v}}$ to compare.
    There are three cases to consider with respect to $\inprod{\u{+}r\v}{\v}$, as follows:\\
    (\romannumeral 1) If $\inprod{\u{+}r\v}{\v}{>}0$, by the triangle inequality we have
    \begin{align}
    \nonumber
    &\ang{\u+r\v}{\v}=\frac{\inprod{\u+r\v}{\v}}{\norm{\u+r\v}\norm{\v}}\ge\frac{\inprod{\u+r\v}{\v}}{({\norm{\u}+\norm{r\v}})\norm{\v}} \\
    \nonumber
    &=\frac{\inprod{\u}{\v}}{\norm{\u}\norm{\v}}+\frac{r\norm{\v}}{\norm{\u}+r\norm{\v}}\Big(1-\frac{\inprod{\u}{\v}}{\norm{\u}\norm{\v}}\Big) \\
    &\ge\frac{\inprod{\u}{\v}}{\norm{\u}\norm{\v}}=\ang{\u}{\v}.
    \end{align}
    (\romannumeral 2) If $\inprod{\u{+}r\v}{\v}{=}0$, $\inprod{\u{+}r\v}{\v}{=}\inprod{\u}{\v}{+}r\norm{\v}^2{=}0$ such that $\ang{\u}{\v}{<}0{=}\ang{\u{+}r\v}{\v}$. \\
    (\romannumeral 3) If $\inprod{\u{+}r\v}{\v}{<}0$, $\inprod{\u}{\v}{<}{-}r\norm{\v}^2{<}0$.
    Let $\u'{=}\u{+}r\v$ and $\v'{=}{-}\v$.
    Because $\inprod{\u'+r\v'}{\v'}{=}{-}\inprod{\u}{\v}{>}0$, as with case (\romannumeral 1) we have $\ang{\u'{+}r\v'}{\v'}{>}\ang{\u'}{\v'}$ and thus $\ang{\u}{{-}\v}{>}\ang{\u{+}r\v}{{-}\v}$, which easily gives $\ang{\u}{\v}{<}\ang{\u{+}r\v}{\v}$.
\end{proof}

\textbf{Second, function $\psi_i$ has a behavioral interpretation.}

By comparing vectors $\vec{s}_i$ and $\vec{s}'_i$ component-wise, there must be less proportion of pure strategy $\pi_{ij}$ used in $\vec{s}_i$ than in $\vec{s}'_i$, if its corresponding component $\varphi_{ij}(\vec{s}_i)$ in vector $\lambda_i(\vec{s}_i)$ is zero.
In other words, $\psi_i$ preferably suppresses the use of ``less profitable than average'' pure strategies and meantime enhances the use of some of the ``more profitable than average'' ones.
Note that this behavioral aspect was mentioned in Nash's existential proof.

\textbf{Third, $\vec{s}'_i$ gives more payoff and less regret than $\vec{s}_i$.}

Theorem~\ref{theo:payoff_increase} and Theorem~\ref{theo:vgs_decrease} show that $\vec{s}'_i$ will always be a better choice in terms of payoff and regret sum, no matter what $\vec{s}_i$ the player unilaterally changes its mixed strategy to, unless $\vec{s}_i$ itself is optimal.

\begin{theorem}
For any $\vec{S}{\in}\mathbb{S}$ and $\vec{s}_i{\in}\mathbb{I}_i$, $p_{i,\vec{S}}(\vec{s}'_i){\ge}p_{i,\vec{S}}(\vec{s}_i)$ where $\vec{s}'_i{=}\psi_{i,\vec{S}}(\vec{s}_i)$.\label{theo:payoff_increase}
\end{theorem}
\begin{proof}
    By Eq. (\ref{eq:payoff_inprod}) we know $p_i(\vec{s}'_i){=}\inprod{\psi_i(\vec{s}_i)}{\vec{v}_i}$ where $\vec{v}_i$ is the player $i$'s vertex payoff vector as in Eq. (\ref{eq:vertex_payoff}). Then by Eq. (\ref{eq:psi}) we have,
    \begin{align}
    \nonumber
    &p_i(\vec{s}'_i)=\Big\langle\frac{\vec{s}_i+r_i\lambda_i(\vec{s}_i)}{1+r_i\abs{\lambda_i(\vec{s}_i)}},\vec{v}_i\Big\rangle \\
    \nonumber
    &=\frac{1}{1+r_i\abs{\lambda_i(\vec{s}_i)}}\Big(\inprod{\vec{s}_i}{\vec{v}_i}+r_i\inprod{\lambda_i(\vec{s}_i)}{\vec{v}_i}\Big) \\
    &=\inprod{\vec{s}_i}{\vec{v}_i}+\frac{\inprod{\lambda_i(\vec{s}_i)}{\vec{v}_i}-\inprod{\vec{s}_i}{\vec{v}_i}\abs{\lambda_i(\vec{s}_i)}}{r_i^{-1}+\abs{\lambda_i(\vec{s}_i)}}.
    \end{align}
    Also by Eq. (\ref{eq:payoff_inprod}) we have $\inprod{\vec{s}_i}{\vec{v}_i}{=}p_i(\vec{s}_i)$. And
    \begin{align}
    \nonumber
    &\inprod{\lambda_i(\vec{s}_i)}{\vec{v}_i}-\inprod{\vec{s}_i}{\vec{v}_i}\abs{\lambda_i(\vec{s}_i)} \\
    \nonumber
    &=\sum_{j}^{g_i}\varphi_{ij}(\vec{s}_i)\bar{p}_i(\pi_{ij};\vec{S})-\sum_{j}^{g_i}p_i(\vec{s}_i)\varphi_{ij}(\vec{s}_i) \\
    \nonumber
    &=\sum_{j}^{g_i}\varphi_{ij}(\vec{s}_i)\Big(\bar{p}_i(\pi_{ij};\vec{S})-p_i(\vec{s}_i)\Big) \\
    &=\sum_{j}^{g_i}\varphi_{ij}(\vec{s}_i)\varphi_{ij}(\vec{s}_i)=\norm{\lambda_i(\vec{s}_i)}^2.\label{eq:increase-payoff-1}
    \end{align}
    Here $\norm{\cdot}$ is $L2$-norm of vector.
    Finally because $\norm{\lambda_i(\vec{s}_i)}^2{\ge}0$ and $r_i^{-1}{+}\abs{\lambda_i(\vec{s}_i)}{\ge}r_i^{-1}{>}0$, we have
    \begin{equation}
    p_i(\vec{s}'_i)=p_i(\vec{s}_i)+\frac{\norm{\lambda_i(\vec{s}_i)}^2}{r^{-1}_i+\abs{\lambda_i(\vec{s}_i)}}\ge{}p_i(\vec{s}_i).
    \label{eq:increase-payoff-2}
    \end{equation}
    And $p_i(\vec{s}'_i){=}p_i(\vec{s}_i)$ if and only if $\abs{\lambda_i(\vec{s}_i)}{=}0$.
\end{proof}

\begin{theorem}
    For any $\vec{S}{\in}\mathbb{S}$ and $\vec{s}_i{\in}\mathbb{I}_i$, $\abs{\lambda_{i,\vec{S}}(\vec{s}'_i)}{\le}\abs{\lambda_{i,\vec{S}}(\vec{s}_i)}$ where $\vec{s}'_i{=}\psi_{i,\vec{S}}(\vec{s}_i)$.\label{theo:vgs_decrease}
\end{theorem}
\begin{proof}
    For each component $\varphi_{ij}(\vec{s})$ in vector $\lambda_i(\vec{s})$ and its counterpart $\varphi_{ij}(\vec{s}'_i)$ in vector $\lambda_i(\vec{s}'_i)$, by Eq. (\ref{eq:varphi}) and Theorem~\ref{theo:payoff_increase} we have $\varphi_{ij}(\vec{s}){\ge}\varphi_{ij}(\vec{s}'_i)$ such that $\sum_j^{g_i}\varphi_{ij}(\vec{s}'_i){\le}\sum_j^{g_i}\varphi_{ij}(\vec{s}_i)$.
\end{proof}

These three aspects makes function $\psi_i$ a natural regret matching of mixed strategies:
if a player $i$ unilaterally and continuously adjusts its strategy by $\psi_i$, intuitively its strategy will keep ``pushing'' towards regret vector as target, the use of ``unprofitable'' pure strategies will keep dropping, and most importantly its payoff will monotonically increase and regret sum will monotonically decrease.
Since a player's payoff is bounded, continuous strategy adjustment by $\psi_i$ would preferably lead to the optimal mixed strategy with maximum payoff and zero regret sum.

The rest of this section considers all players adjusting their strategies simultaneously and continuously.

On the game with given $\vec{S}$, if all $n$ players simultaneously adjust strategies by their own $\psi_i$, there will form a new $\vec{S}'$. Let $\vec{S}'{=}\Psi(\vec{S})$ and we have

\begin{equation}
\Psi(\vec{S})=\Big(\psi_{1,\vec{S}}(\vec{s}_1),\ldots,\psi_{i,\vec{S}}(\vec{s}_i),\ldots,\psi_{n,\vec{S}}(\vec{s}_n)\Big).\label{eq:Psi}
\end{equation}

Note that $\Psi(\vec{S}){\in}\mathbb{S}$ for any $\vec{S}{\in}\mathbb{S}$ because $\psi_i(\vec{s}_i){\in}\mathbb{I}_i$ for any $\vec{s}_i{\in}\mathbb{I}_i$.
And all players' continuous strategy adjustment can be described by an iterated function $\Psi{\circ}\Psi{\circ}{\cdots}{\circ}\Psi(\vec{S})$, denoted by $\Psi^t(\vec{S})$ with $t{\in}\mathbb{N}$ being the times of iterations.
Given $\vec{S}{=}\vec{S}_0$, those iterated functions $\Psi^t(\vec{S})$ with $t{=}0,1,2,\ldots$ give rise to an infinite sequence of points in $\mathbb{S}$ for our study:

\begin{align}
\label{eq:sequence}
\nonumber
\vec{S}_0&=\Psi^0(\vec{S}_0),\\
\nonumber
\vec{S}_1&=\Psi^1(\vec{S}_0)=\Psi(\vec{S}_0),\\
\nonumber
\vec{S}_2&=\Psi^2(\vec{S}_0)=\Psi(\vec{S}_1),\\
\nonumber
&\cdots\\
\vec{S}_{t}&=\Psi^t(\vec{S}_0)=\Psi(\vec{S}_{t-1}),\\
\nonumber
&\cdots
\end{align}

Here $\Psi^0$ denotes the identity function.
This sequence is referred to as $(\vec{S}_{t})_{\mathbb{N}}$ implicitly indexed by $\vec{S}_0$ and $\Psi$.
In fact, the tuple $[\Psi,\vec{S}_0]$ uniquely determines a sequence $(\vec{S}_{t})_{\mathbb{N}}$; we say $[\Psi,\vec{S}_0]$ generates $(\vec{S}_{t})_{\mathbb{N}}$.
Sequence $(\vec{S}_{t})_{\mathbb{N}}$ describes an infinite dynamics of players adjusting their strategies in parallel.
Specifically, at the $t$\textsuperscript{th} iteration, each player $i$ adjusts its strategy by substituting $\vec{s}_i$ in $\vec{S}_{t-1}$ with $\psi_i(\vec{s}_i)$ so as to collectively form a new $n$-tuple $\vec{S}_{t}{=}\Psi(\vec{S}_{t-1})$;
at the $(t{+}1)$\textsuperscript{th} iteration $\vec{S}_{t+1}$ is formed by adjusting $\vec{S}_{t}$ in the same manner; and so forth.

Now we might hope that each and every player still has its payoff monotonically increase and its regret sum monotonically decrease in the process of simultaneous strategy adjustment.
In that case, the infinite sequence $(\vec{S}_{t})_{\mathbb{N}}$ must tend towards an \eqpt{}.
Moreover, it would be ideal if the overall regret sum of all players, i.e. $\sum^n_i\abs{\lambda_i(\vec{s}_i)}$, converged to zero along $(\vec{S}_{t})_{\mathbb{N}}$ such that $\lim_{\infty}\sum^n_i\abs{\lambda_i(\vec{s}_i)}^{(t)}{=}0$ and thus $\lim_{\infty}\abs{\lambda_i(\vec{s}_i)}^{(t)}{=}0$ for each player $i$.
In that case, by definition $\lim_{\infty}\vec{S}_{t}$ must be an \eqpt{}. Denote it by $\vec{S}^+$ and we must have $\Psi(\vec{S}^+){=}\vec{S}^+$.
Therefore, \eqpt{} $\vec{S}^+{=}\lim_{\infty}\vec{S}_{t}$ is also a fixed point~\cite{fixedpoint} of function $\Psi$;
the sequence $(\vec{S}_{t})_{\mathbb{N}}$ generated by $[\Psi,\vec{S}_0]$ is essentially an outcome of the process of fixed point iteration, which can be formulated to $\lim_{\infty}\Psi{\circ}\Psi{\circ}{\cdots}{\circ}\Psi(\vec{S}_0){=}\lim_{\infty}\Psi^t(\vec{S}_0){=}\vec{S}^+$.
For our purpose we assume that, with proper finite iterations, $\vec{S}^{*}{=}\vec{S}_t{=}\Psi^t(\vec{S}_0)$ can be a reasonable approximate of the true \eqpt{} $\vec{S}^+$ such that $\vec{S}^{*}{\approx}\vec{S}^+$, and the accuracy of such \eqpt{} approximation can be measured by the overall regret sum of all players.

Only this hope comes with huge caveat.
As it turns out in Section~\ref{section:accuracy}, the presumption on the monotonicity of payoff and regret sum for all players along $(\vec{S}_{t})_{\mathbb{N}}$ oftentimes doesn't hold true.
Instead, due to simultaneous strategy adjustment each player's regret sum might fluctuate along $(\vec{S}_{t})_{\mathbb{N}}$ with a tendency to diminish towards zero.
Nevertheless, we will see that this tendency can still serve the purpose of approximating an \eqpt{}.
Notice that in FIG~\ref{fig:angle_close_up}, the parameter $r_i$ in Eq. (\ref{eq:psi}) determines the rate of angle by which vector $\vec{s}_i$ and regret vector $\lambda_i(\vec{s}_i)$ ``close up''.
Here we call $r_i$ adjustment rate.
To counter the influence from simultaneous strategy adjustment, as we will see in next section ``infinitesimal'' adjustment rates are used for all players, in the hope that each player's strategy adjustment could be treated as negligible by its opponents and thus their simultaneous strategy adjustment could be treated as approximately unilateral.




%
%

We should visualize the sequence $(\vec{S}_{t})_{\mathbb{N}}$ of Eq. (\ref{eq:sequence}) as players' perpetual searching for the equilibrium point of non{-}cooperative game.
In this searching, players don't have to be aware of the equilibrium point per se.
Instead, each player is constantly chasing after its regret vector $\lambda_i(\vec{s}_i)$ in the hope of reducing its regret sum $\abs{\lambda_i(\vec{s}_i)}$ to zero.
Increasing payoff and meantime reducing regret sum is the players' sole incentive and purpose, and regret vector is the players' sole information and target for their strategies adjustment.
Regret vector is dynamic along iterations, of course.
At any iteration, given the current mixed strategies of its opponents, each player can calculate its regret vector by evaluating the payoff on its pure strategies without any across-iteration retrospective considerations.
In this searching, there is no place for the mediator beyond players; no inter-player information exchange other than the current use of mixed strategies.

\section{Game examples and visualizations\label{section:2p_game}}

In this section we will verify our approach of \eqpt{} approximation proposed in last section.
Examples of both two-person game and general $n$-person game will be tested, although the former ones are our main consideration due to their simplicity and sufficiency in revealing the important aspects of \eqpt{} approximation.

For the two-person game we have $n{=}2$ players, say the row player with $k$ pure strategies and the column player with $m$ pure strategies.
Accordingly, the $n$-tuple $\vec{S}$ is specialized to two-tuple $(\vec{s}^{1{\times}k}_r,\vec{s}^{1{\times}m}_c)$ of vectors, and thus the tuple $[\Psi,\vec{S}_0]$ is specialized to $[(\vec{A},\vec{B}),(r_r, r_c),(\vec{s}_r,\vec{s}_c)_0]$ since function $\Psi$ is determined by the two players' payoff matrices $(\vec{A}^{k{\times}m},\vec{B}^{k{\times}m})$ and their adjustment rates $(r_r,r_c)$.
Recall that given any $\vec{S}_0{\in}\mathbb{S}$ the tuple $[\Psi,\vec{S}_0]$ generates an infinite sequence $(\vec{S}_{t})_{\mathbb{N}}$ correspondingly.
Then we can implement the fixed point iteration in last section to Algorithm~\ref{alg:two-person}\footnote{Source code and demos can be found at \emph{https://github.com/lansiz/eqpt}.}, which takes $[(\vec{A},\vec{B}),(r_r, r_c),(\vec{s}_r,\vec{s}_c)_0]$ as input, and for our study outputs an approximate \eqpt{} $(\vec{s}_r,\vec{s}_c)^{*}$ as $\vec{S}^{*}$ and its regret sum $(R_r,R_c)^{*}$, a finite sequence $\big((\vec{s}_r,\vec{s}_c)_{t}\big)_{T}$ in correspondence to $(\vec{S}_{t})_{\mathbb{N}}$, and a finite sequence $\big((R_r,R_c)_{t}\big)_{T}$ of regret sum.
In Algorithm~\ref{alg:two-person}, the $\max$ operation compares vectors component-wise as opposed to Eq. (\ref{eq:varphi}).
Obviously, given $T$ the computational complexity of Algorithm~\ref{alg:two-person} is under $O(g^2)$ with $g{=}{\max}\{k,m\}$.

\begin{algorithm}[H]
\begin{algorithmic}[1]
\item[\textbf{Input:}]{payoff bimatrix $(\vec{A}^{k{\times}m},\vec{B}^{k{\times}m})$, adjustment rates $(r_r,r_c)$, initial two-tuple of mixed strategies $(\vec{s}^{1{\times}k}_r,\vec{s}^{1{\times}m}_c)$ and iterations $T$.}
\item[\textbf{Output:}]{
approximate \eqpt{} $(\vec{s}_r,\vec{s}_c)^{*}$ and its regret sum $(R_r,R_c)^{*}$; sequence $\big((\vec{s}_r,\vec{s}_c)_{t}\big)_{T}$; sequence $\big((R_r,R_c)_{t}\big)_{T}$.}
\State\textbf{for} $t{=}1$ \textbf{to} $T$ \textbf{do}
\State\hspace{0.5cm}vertices payoff: $\vec{v}_r=\vec{A}\vec{s}^{\intercal}_c$, $\vec{v}_c=\vec{B}^{\intercal}\vec{s}^{\intercal}_r$.
\State\hspace{0.5cm}payoff: $p_r=\vec{s}_r\vec{v}_r$, $p_c=\vec{s}_c\vec{v}_c$.
\State\hspace{0.5cm}regret vector: $\vec{R}_r{=}{\max}\{\vec{0},\vec{v}_r{-}p_r\}$, $\vec{R}_c{=}{\max}\{\vec{0},\vec{v}_c{-}p_c\}$.
\State\hspace{0.5cm}regret sum: $R_r{=}\abs{\vec{R}_r}$, $R_c{=}\abs{\vec{R}_c}$.
\State\hspace{0.5cm}\textbf{if} overall regret sum $R_r{+}R_c$ is a new minimum \textbf{then}
\State\hspace{0.5cm}\hspace{0.5cm}update $(\vec{s}_r,\vec{s}_c)^{*}$ with $(\vec{s}_r,\vec{s}_c)$.
\State\hspace{0.5cm}\hspace{0.5cm}update $(R_r,R_c)^{*}$ with $(R_r,R_c)$.
\State\hspace{0.5cm}\textbf{endif}
\State\hspace{0.5cm}append $(\vec{s}_r,\vec{s}_c)$ into its sequence.
\State\hspace{0.5cm}append $(R_r,R_c)$ into its sequence.
\State\hspace{0.5cm}adjust mixed strategies: $\vec{s}_r{=}\frac{\vec{s}_r+r_r\vec{R}_r}{1+r_rR_r}$, $\vec{s}_c{=}\frac{\vec{s}_c+r_c\vec{R}_c}{1+r_cR_c}$.
\State\textbf{end for}
\end{algorithmic}
\caption{\eqpt{} approximation of two-person game.\label{alg:two-person}}
\end{algorithm}

Given the output of Algorithm~\ref{alg:two-person}, we can plot the true \eqpt{} $(\vec{s}_r,\vec{s}_c)^{+}$, the approximate \eqpt{} $(\vec{s}_r,\vec{s}_c)^{*}$ and the sequence $\big((\vec{s}_r,\vec{s}_c)_{t}\big)_{T}$ on $\mathbb{R}^2$ or $\mathbb{R}^3$ to provide an intuitive visualization of $\big((\vec{s}_r,\vec{s}_c)_{t}\big)_{\mathbb{N}}$ approximating \eqpt{}.
Specifically, for the $3{\times}3$ two-person games with $k{=}3$ and $m{=}3$, we split $\big((\vec{s}_r,\vec{s}_c)_{t}\big)_{T}$ into two sequences of $(\vec{s}_{t})_{T}$, and then transform each into a sequence of $\big((x,y)_{t}\big)_{T}$ for plotting on $\mathbb{R}^2$ plane by defining

\begin{equation}
(x,y)_{t}=\vec{s}_{t}
\begin{bmatrix}
    0 & 0 \\
    \frac{\sqrt{2}}{2} & \frac{\sqrt{6}}{2} \\
    \sqrt{2} & 0
\end{bmatrix}
.
\label{eq:barycentric2cartesian}
\end{equation}

Note that Eq. (\ref{eq:barycentric2cartesian}) transforms the simplex in $\mathbb{R}^3$ into a subset of $\mathbb{R}^2$ enclosed by an equilateral triangle.
The following are our observations on the test results of four typical $3{\times}3$ games:
\begin{itemize}
\item{Game \game{3X3-1eq1sp} has one unique \eqpt{}, which uses one single pure strategy. As shown in FIG~\ref{fig:3X3-1eq1sp}, $\big((\vec{s}_r,\vec{s}_c)_{t}\big)_{T}$ converges straight towards the \eqpt{}. }
\item{Game \game{3X3-1eq2sp} has one unique \eqpt{}, which uses two pure strategies. As shown in FIG~\ref{fig:3X3-1eq2sp}, $\big((\vec{s}_r,\vec{s}_c)_{t}\big)_{T}$ converge towards the \eqpt{} with oscillation. }
\item{Game \game{3X3-2eq2sp} has two \eqpt{}s, each of which uses two pure strategies. As shown in FIG~\ref{fig:3X3-2eq2sp}, $\big((\vec{s}_r,\vec{s}_c)_{t}\big)_{T}$ converges with oscillation towards one of the two \eqpt{}s dependent on the initial strategies. }
\item{Game \game{3X3-1eq3sp} has one unique \eqpt{}, which uses three pure strategies. As shown in FIG~\ref{fig:3X3-1eq3sp}, $\big((\vec{s}_r,\vec{s}_c)_{t}\big)_{T}$ doesn't converge towards the \eqpt{} but develops into a seemingly perfect circle away from the true \eqpt{}. The last position of $\big((\vec{s}_r,\vec{s}_c)_{t}\big)_{T}$ could be any point in the circle. }

\end{itemize}

\begin{figure}[h]
\centering
\includegraphics[width=\linewidth]{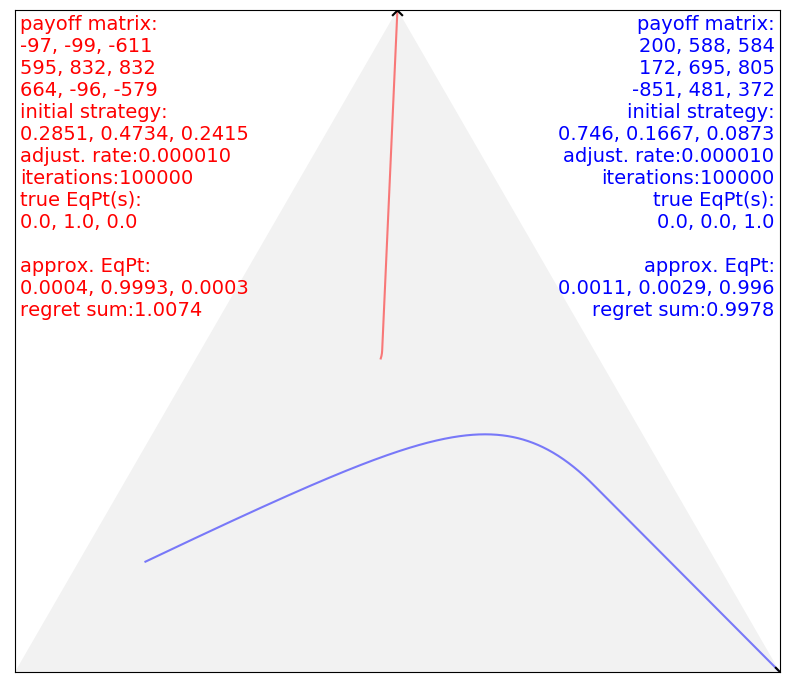}
\caption{\label{fig:3X3-1eq1sp}
    The ``strategy paths'' on the simplex of $\mathbb{R}^3$ for game \game{3X3-1eq1sp}.
    The two black crosses combine to represent the true \eqpt{}.
    The input and output of game are annotated.
}
\end{figure}

\begin{figure}[h]
\centering
\includegraphics[width=\linewidth]{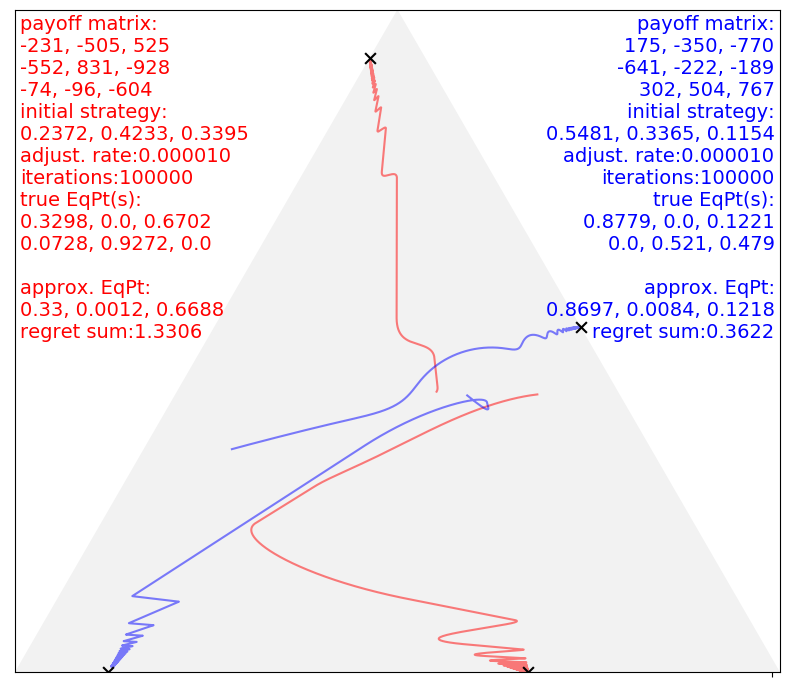}
\caption{\label{fig:3X3-2eq2sp}
    The ``strategy paths'' for game \game{3X3-2eq2sp}.
    The approximate \eqpt{} annotated corresponds to the true \eqpt{} at the base line of equilateral triangle.
}
\end{figure}

\begin{figure}[h]
\centering
\includegraphics[width=\linewidth]{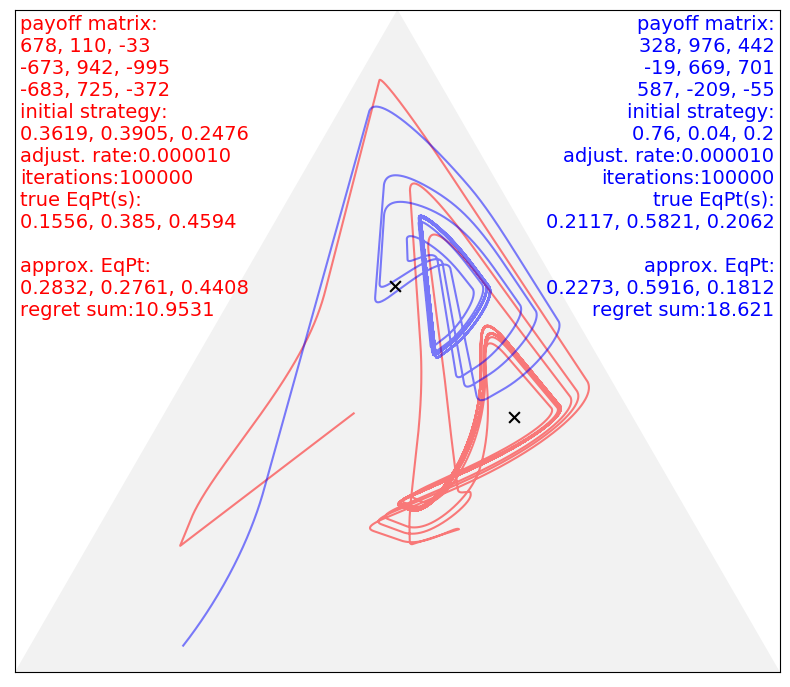}
\caption{\label{fig:3X3-1eq3sp}
     The ``strategy paths'' for game \game{3X3-1eq3sp}.
}
\end{figure}

\begin{figure}[h]
\centering
\includegraphics[width=\linewidth]{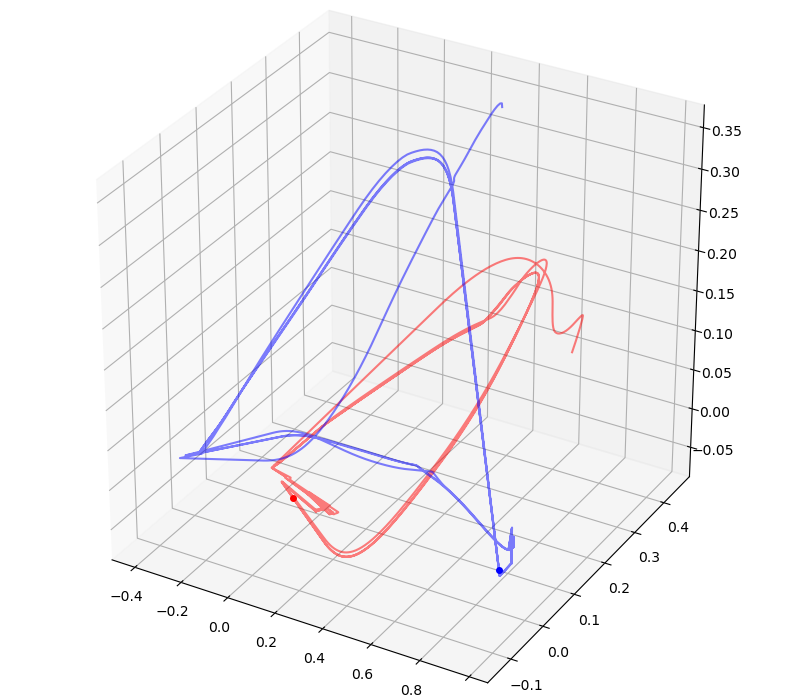}
\caption{\label{fig:60X40-disconv}
    The ``strategy paths'' for a $60{\times}40$ game.
    The blue and red dots represent the last positions of strategies.
}
\end{figure}

And for the general two-person games with $k{>}3$ and $m{>}3$, we split $\big((\vec{s}_r,\vec{s}_c)_{t}\big)_{T}$ into two sequences of $(\vec{s}_{t})_{T}$, and then with PCA (Principal Component Analysis for dimension reduction~\cite{pca}) transform each into a sequence of $\big((x,y,z)_{t}\big)_{T}$ for plotting in $\mathbb{R}^3$ space.
The following are our observations on the test results of two typical $60{\times}40$ games:

\begin{itemize}
\item{FIG~\ref{fig:60X40-conv} shows a $60{\times}40$ game whose $\big((\vec{s}_r,\vec{s}_c)_{t}\big)_{T}$ converges towards an \eqpt{}. }
\item{FIG~\ref{fig:60X40-disconv} shows a $60{\times}40$ game whose $\big((\vec{s}_r,\vec{s}_c)_{t}\big)_{T}$ develops into a perfect circle. }
\end{itemize}

For the two-person game, its \eqpt{}s are determined by the bimatrix $(\vec{A},\vec{B})$, independent of the initial $(\vec{s}_r,\vec{s}_c)_{0}$.
Then for the $3{\times}3$ game examples mentioned previously, we test them with many random $(\vec{s}_r,\vec{s}_c)_{0}$ and observe that there appear to be two kinds of \eqpt{}s, ``attractor'' \eqpt{} and ``repellor'' \eqpt{}, as follows:

\begin{figure}[h]
\centering
\includegraphics[width=\linewidth]{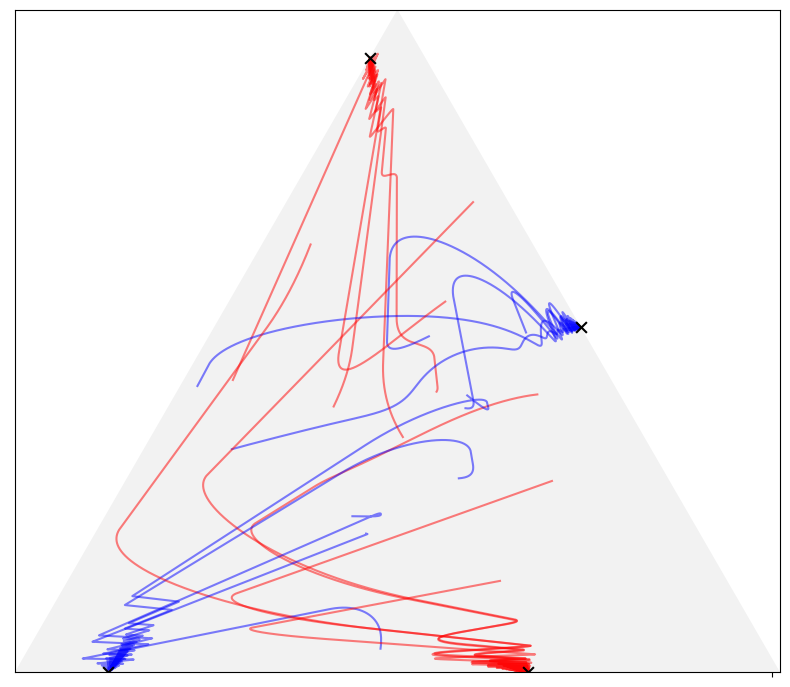}
\caption{\label{fig:3X3-2eq2sp-attractor}
    Both \eqpt{}s of game \game{3X3-2eq2sp} are ``attractors''.
    Ten pairs of random ``strateg paths'' are depicted.
}
\end{figure}

\begin{figure}[h]
\centering
\includegraphics[width=\linewidth]{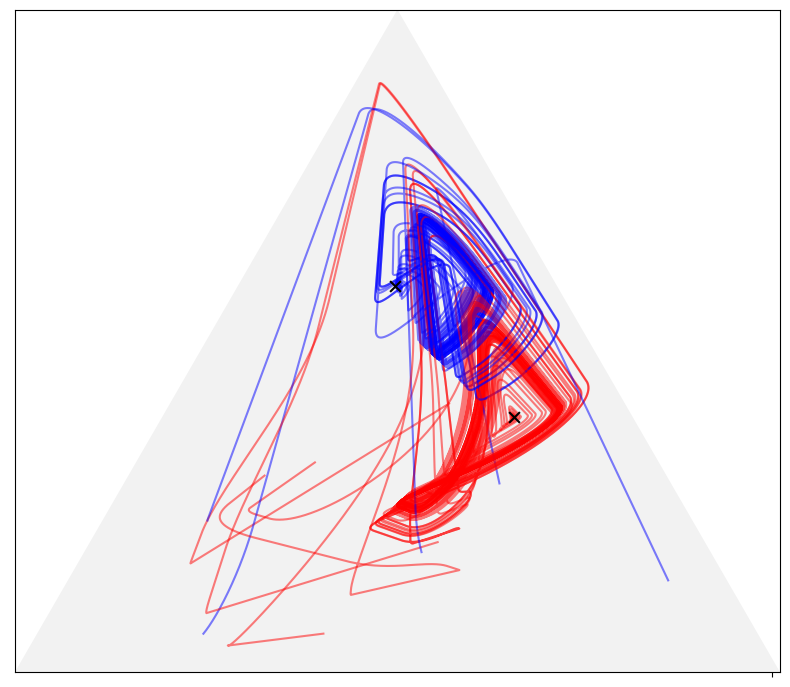}
\caption{\label{fig:3X3-1eq3sp-repellor}
    The \eqpt{} of game \game{3X3-1eq3sp} is a ``repellor''.
    Ten pairs of random ``strateg paths'' are depicted, five of which are started from closely near the true \eqpt{}.
}
\end{figure}

\begin{itemize}
\item{For game \game{3X3-1eq2sp} FIG~\ref{fig:3X3-1eq2sp-attractor} shows that, any $\big((\vec{s}_r,\vec{s}_c)_{t}\big)_{T}$ is ``attracted'' towards the unique \eqpt{}.}
\item{For game \game{3X3-2eq2sp} FIG~\ref{fig:3X3-2eq2sp-attractor} shows that, any $\big((\vec{s}_r,\vec{s}_c)_{t}\big)_{T}$ is ``attracted'' towards one of the two \eqpt{}s.}
\item{For game \game{3X3-1eq3sp} FIG~\ref{fig:3X3-1eq3sp-repellor} shows that, any $\big((\vec{s}_r,\vec{s}_c)_{t}\big)_{T}$ is ''repelled'' from the \eqpt{} unless $(\vec{s}_r,\vec{s}_c)_{0}$ is exactly the \eqpt{}.}
\end{itemize}

We also test the examples of $n$-person game by extending Algorithm~\ref{alg:two-person} from bivariate case to general multivariate case, which can be better explained by the source code.
And we observe the convergence or disconvergence of $(\vec{S}_{t})_{T}$ as with the two-person games.
FIG.~\ref{fig:nP-path} shows an example for the convergence case.
The computational complexity of this extended algorithm is under $O(n^2g^{n})$ with $n$ being the number of players and $g{=}{\max}\{g_1,g_2,\ldots,g_n\}$, meaning that it is polynomial time when $n$ is given.

\section{\eqpt{} approximation accuracy\label{section:accuracy}}

In last section, we observe that for some games sequence $\big((\vec{s}_r,\vec{s}_c)_{t}\big)_{T}$ seemingly converges towards \eqpt{}, e.g., \game{3X3-1eq1sp}, \game{3X3-1eq2sp} and \game{3X3-2eq2sp}.
In contrast, for games such as \game{3X3-1eq3sp}, sequence $\big((\vec{s}_r,\vec{s}_c)_{t}\big)_{T}$ clearly doesn't converge at all; instead, it evolves into a perpetual cyclic path away from \eqpt{}.
Roughly speaking, the convergence or disconvergence of $\big((\vec{s}_r,\vec{s}_c)_{t}\big)_{T}$ affects the accuracy of \eqpt{} approximation.
In this section, we will look into the approximation accuracy from two numerical perspectives.
One is to examine sequence $\big((R_r,R_c)_{t}\big)_{T}$ since regret sum is a measure of approximation accuracy as discussed in Section~\ref{section:fpi_eqpt}, and the other one is to revisit the convergence of $\big((\vec{s}_r,\vec{s}_c)_{t}\big)_{T}$ with metric.
And we will show the connection between them.

\begin{figure*}
\centering
\begin{subfigure}[b]{0.485\textwidth}
    \centering
    \includegraphics[width=\textwidth]{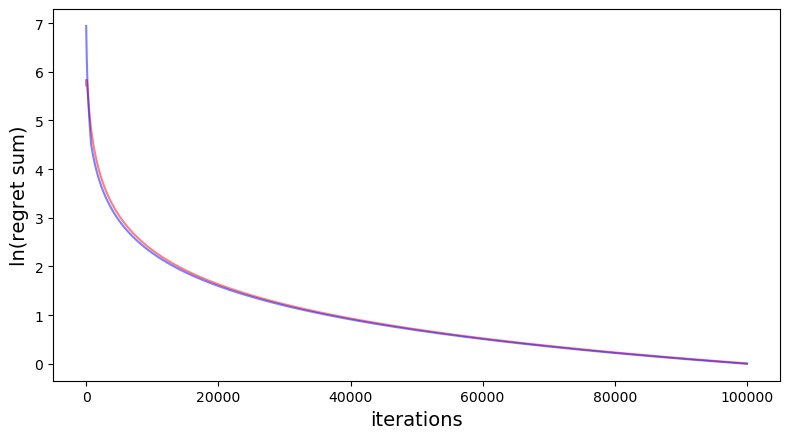}
    \caption[]{Game \game{3X3-1eq1sp}.\label{fig:3X3-1eq1sp-vgs}}
\end{subfigure}
\quad
\begin{subfigure}[b]{0.485\textwidth}
    \centering
    \includegraphics[width=\textwidth]{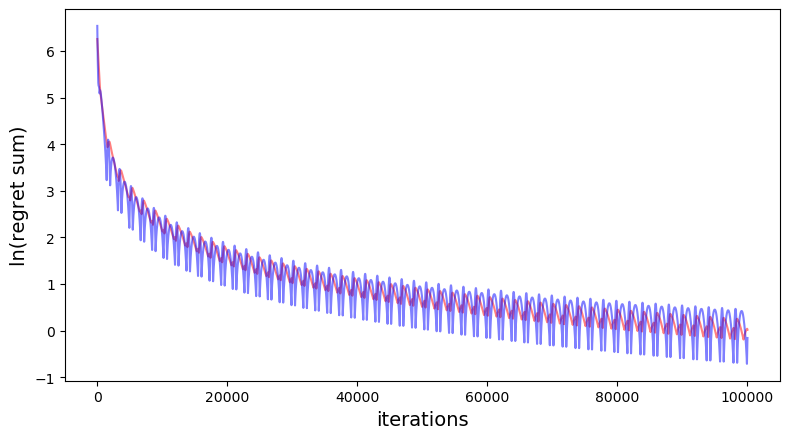}
    \caption[]{Game \game{3X3-1eq2sp}.\label{fig:3X3-1eq2sp-vgs}}
\end{subfigure}
\vskip\baselineskip
\begin{subfigure}[b]{0.485\textwidth}
    \centering
    \includegraphics[width=\textwidth]{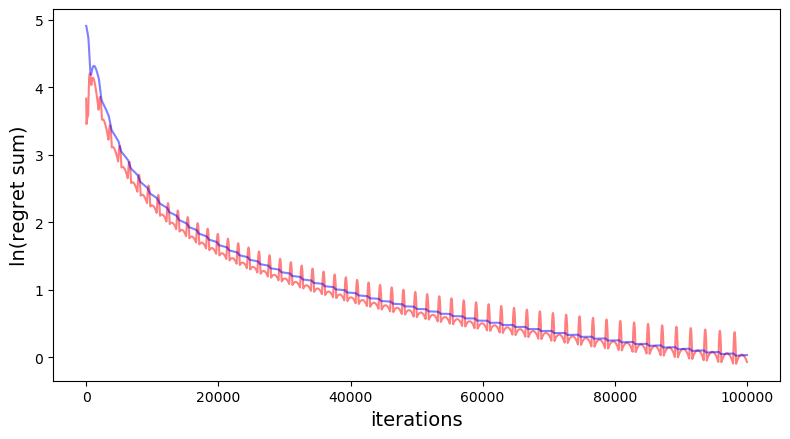}
    \caption[]{Game \game{3X3-2eq2sp}.\label{fig:3X3-2eq2sp-vgs}}
\end{subfigure}
\quad
\begin{subfigure}[b]{0.485\textwidth}
    \centering
    \includegraphics[width=\textwidth]{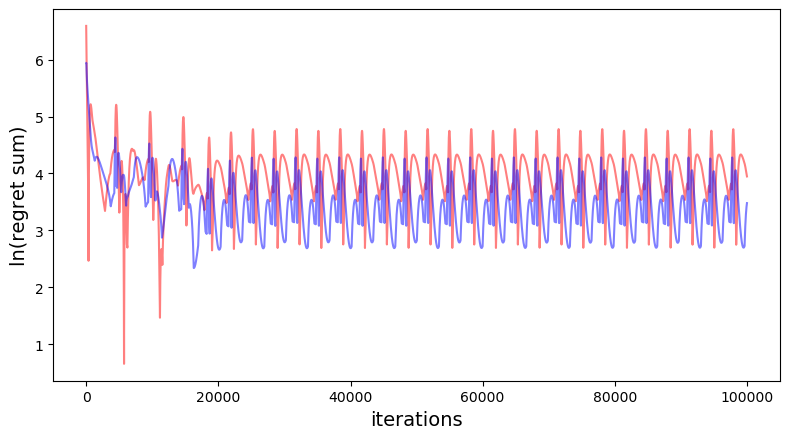}
    \caption[]{Game \game{3X3-1eq3sp}.\label{fig:3X3-1eq3sp-vgs}}
\end{subfigure}
\caption[]{
    The output sequences $\big((R_r,R_c)_{t}\big)_{T}$ of Algorithm~\ref{alg:two-person} for the four typical $3{\times}3$ games in last section.
}
\label{fig:examples-vgs}
\end{figure*}

First, by taking the $3{\times}3$ games in last section for example, we have the following observations on their sequences $\big((R_r,R_c)_{t}\big)_{T}$ plotted in FIG~\ref{fig:examples-vgs}:

\begin{itemize}
\item{
    As opposed to the presumption in Section~\ref{section:fpi_eqpt}, regret sum oftentimes doesn't monotonically decrease.
    Instead, due to the inter-player influence from simultaneous strategy adjustment, regret sum fluctuates with a decreasing tendency to reach new minimum.
    Game \game{3X3-1eq1sp} barely has regret sum fluctuation, whereas game \game{3X3-1eq3sp} has intensive regret sum fluctuation.
}
\item{
    There is different degree of periodicity in regret sum fluctuation for different games.
    Regret sum fluctuation of game \game{3X3-1eq3sp} exhibits strong periodicity and yet weak decreasing tendency.
}
\item{
    The strong periodicity of regret sum sequence in game \game{3X3-1eq3sp}, caused by the ``repellor'' \eqpt{}, severely undermines the accuracy of \eqpt{} approximation.
}
\item{
    It is painfully slow for regret sum sequence to converge towards zero in that regret sum decreases slowly at the latter iterations, even for \game{3X3-1eq1sp}.
    That is partly because, as implied by Eq. (\ref{eq:increase-payoff-2}), the ratio of payoff increment to regret sum is getting smaller as iteration goes.
}
\end{itemize}

Now it can be concluded that the inter-player influence has a significant impact on the \eqpt{} approximation accuracy which otherwise could, given time, go on to perfection.
In addition, FIG~\ref{fig:nperson-vgs} shows the decreasing tendency of regret sum for a five-person game.

\begin{figure}[h]
\centering
\includegraphics[width=\linewidth]{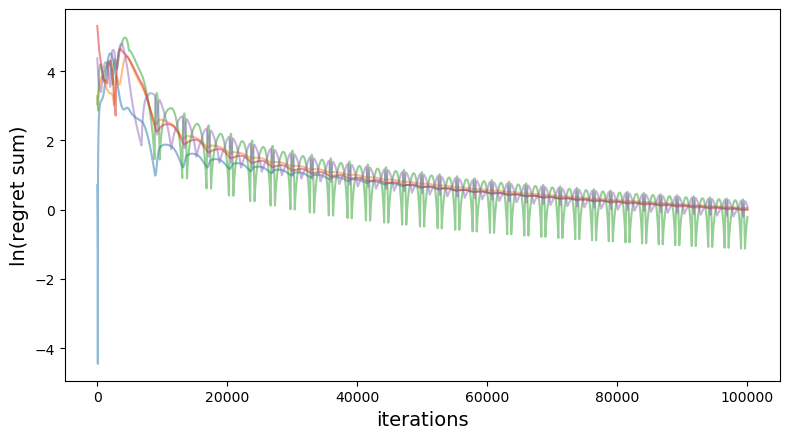}
\caption{\label{fig:nperson-vgs}
The regret sum sequences of a five-person game with players using two, three, four, five and six pure strategies.
}
\end{figure}

Next, for the $3{\times}3$ games in last section, we will reexamine their output sequences $\big((\vec{s}_r,\vec{s}_c)_{t}\big)_{T}$ by introducing a metric on $\mathbb{S}$.
Generally, metric is necessary for the definitions of limit, convergence and contractiveness of sequence $(\vec{S}_{t})_{\mathbb{N}}$.
Following the definitions in Section~\ref{section:fpi_eqpt}, let us have a set $\mathbb{S}$ and a function $\Psi:\mathbb{S}{\rightarrow}\mathbb{S}$ of Eq. (\ref{eq:Psi}).
From $\mathbb{S}$ we can derive a complete metric space $(\mathbb{S},d)$ with the metric function

\begin{equation}
d(\vec{X},\vec{Y})=\sum^n_i\norm{\vec{x}_{i}-\vec{y}_{i}}.\label{eq:metric-function-1}
\end{equation}

Here $\vec{x}_{i}$ and $\vec{y}_{i}$ are the $i$\textsuperscript{th} items of $\vec{X}{\in}\mathbb{S}$ and $\vec{Y}{\in}\mathbb{S}$ respectively.
Assume that $\Psi:\mathbb{S}{\rightarrow}\mathbb{S}$ is a contraction mapping on metric space $(\mathbb{S},d)$. That is, there exists a $q{\in}(0,1)$ such that for any $\vec{X},\vec{Y}{\in}\mathbb{S}$

\begin{equation}
d\big(\Psi(\vec{X}),\Psi(\vec{Y})\big)\le{}qd(\vec{X},\vec{Y}).\label{eq:contraction-property}
\end{equation}

Then by Banach fixed point theorem~\cite{fixedpoint}, the function $\Psi$ must admit one unique fixed point $\vec{S}^+{\in}\mathbb{S}$ such that $\Psi(\vec{S}^+){=}\vec{S}^+$ and $\lim_{\infty}\Psi^t(\vec{S}_0){=}\vec{S}^+$ for any $\vec{S}_0{\in}\mathbb{S}$.
Apparently the contraction condition of Eq. (\ref{eq:contraction-property}) is too strong for the convergence of $(\vec{S}_{t})_{\mathbb{N}}$, since it precludes the sequence's possible convergence towards more than one \eqpt{}, which is an important case of our interest as shown in last section.
Instead, let us consider the contraction property of a specific sequence $(\vec{S}_{t})_{\mathbb{N}}$ generated by the given tuple $[\Psi,\vec{S}_0]$.
From $(\vec{S}_{t})_{\mathbb{N}}$ we can derive an infinite real sequence $(\dot{d}_{t})_{\mathbb{N}_+}$ by defining

\begin{equation}
\dot{d}_{t}=d(\vec{S}_{t-1},\vec{S}_{t})=d\big(\vec{S}_{t-1},\Psi(\vec{S}_{t-1})\big).
\end{equation}

And from sequence $(\dot{d}_{t})_{\mathbb{N}_+}$ we can further derive an infinite real sequence $(\dot{q}_{t})_{\mathbb{N}_+}$ by defining

\begin{equation}
\dot{q}_{t}=\frac{\dot{d}_{t+1}}{\dot{d}_{t}}=\frac{d(\vec{S}_{t},\vec{S}_{t+1})}{d(\vec{S}_{t-1},\vec{S}_{t})}=\frac{d\big(\Psi(\vec{S}_{t-1}),\Psi(\vec{S}_{t})\big)}{d\big(\vec{S}_{t-1},\vec{S}_{t}\big)}.
\end{equation}

\begin{figure*}
\centering
\begin{subfigure}[b]{0.485\textwidth}
    \centering
    \includegraphics[width=\textwidth]{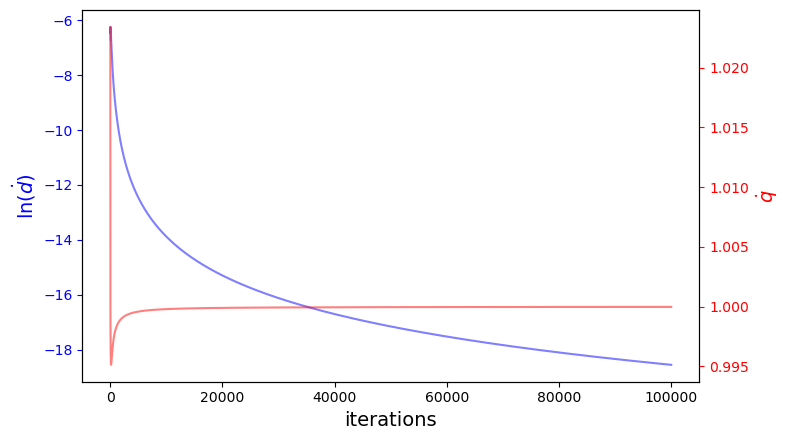}
    \caption[]{Game \game{3X3-1eq1sp}.\label{fig:3X3-1eq1sp-metric}}
\end{subfigure}
\quad
\begin{subfigure}[b]{0.485\textwidth}
    \centering
    \includegraphics[width=\textwidth]{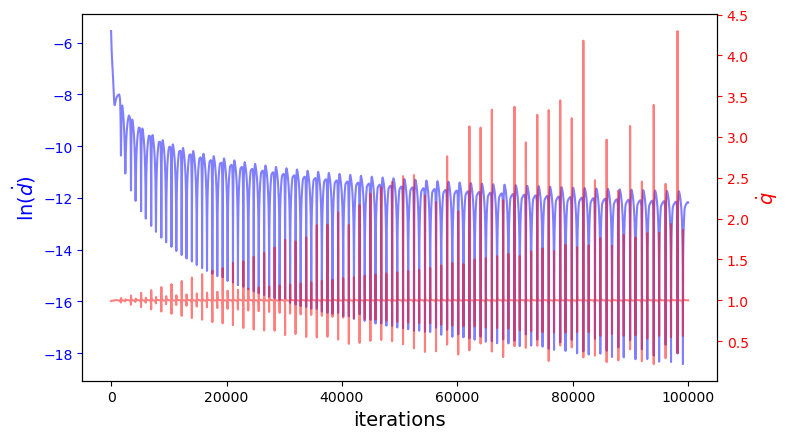}
    \caption[]{Game \game{3X3-1eq2sp}.\label{fig:3X3-1eq2sp-metric}}
\end{subfigure}
\vskip\baselineskip
\begin{subfigure}[b]{0.485\textwidth}
    \centering
    \includegraphics[width=\textwidth]{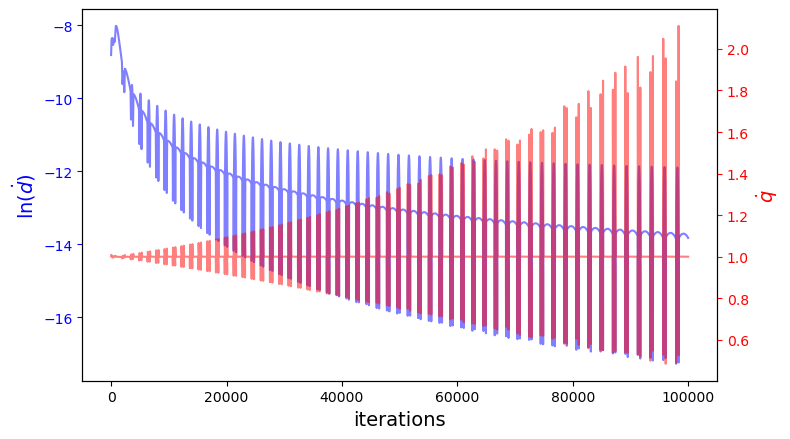}
    \caption[]{Game \game{3X3-2eq2sp}.\label{fig:3X3-2eq2sp-metric}}
\end{subfigure}
\quad
\begin{subfigure}[b]{0.485\textwidth}
    \centering
    \includegraphics[width=\textwidth]{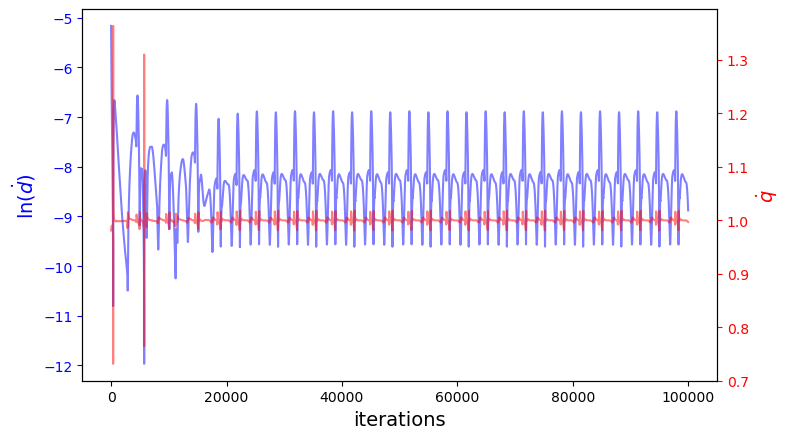}
    \caption[]{Game \game{3X3-1eq3sp}.\label{fig:3X3-1eq3sp-metric}}
\end{subfigure}
\caption[]{
    The derived sequences $(\dot{d}_{t})_{T-1}$ in blue and $(\dot{q}_{t})_{T-2}$ in red for the four typical $3{\times}3$ games in last section.}
\label{fig:examples-metric}
\end{figure*}

If $(\vec{S}_{t})_{\mathbb{N}}$ is a contractive sequence ideally, i.e. $0{<}\dot{q}_{t}{<}1$ for any $t{\in}\mathbb{N}_+$, it must be an Cauchy sequence and converge to an \eqpt{}.
That is, $(\vec{S}_{t})_{\mathbb{N}}$ is a process of fixed point iteration.
Meanwhile, the sequence $(\dot{d}_{t})_{\mathbb{N}_+}$ must monotonically decrease and converge to zero.
From the $\big((\vec{s}_r,\vec{s}_c)_{t}\big)_{T}$ sequences generated by the $3{\times}3$ games of last section, we derive and depict their $(\dot{d}_{t})_{T-1}$ and $(\dot{q}_{t})_{T-2}$ in FIG~\ref{fig:examples-metric} to show that $(\dot{d}_{t})_{T-1}$ sequences don't monotonically decrease or converge to zero, meaning that neither $\big((\vec{s}_r,\vec{s}_c)_{t}\big)_{T}$ sequences nor their $\Psi$ functions are contractive on $(\mathbb{S},d)$.
Nevertheless, by comparing FIG~\ref{fig:examples-vgs} and FIG~\ref{fig:examples-metric}, we observe that there is strong coherence among the sequences $\big((R_r,R_c)_{t}\big)_{T}$, $(\dot{d}_{t})_{T-1}$ and $(\dot{q}_{t})_{T-2}$ of the same game, as follows:

\begin{itemize}
\item{
    As with regret sum, $\dot{d}$ oftentimes fluctuates with a decreasing tendency to reach new minimum.
    $\dot{d}$ of game \game{3X3-1eq1sp} exhibits almost no fluctuation, whereas $\dot{d}$ of game \game{3X3-1eq3sp} exhibits strong periodicity.
}
\item{
    As with regret sum, in game \game{3X3-1eq3sp} the strong periodicity of $(\dot{d}_{t})_{T-1}$ could have hurt the accuracy of \eqpt{} approximation.
}
\item{
    $\dot{q}$ fluctuates around $1$ at the latter iterations, which explains why $(\dot{d}_{t})_{T-1}$ and regret sum sequence converge at such a low speed.
}
\end{itemize}

Therefore, $\dot{d}$ can be seen as an another overall measure of \eqpt{} approximation accuracy in addition to overall regret sum and to improve accuracy is to reduce $\dot{d}$, which in return justifies our definition of metric function in Eq. (\ref{eq:metric-function-1}).
Trivially, there is an alternative metric function and applying it gives us observations similar to those with Eq. (\ref{eq:metric-function-1}).
Here is the alternative~\cite{topology}:

\begin{equation}
d(\vec{X},\vec{Y})=\max_i^n\norm{\vec{x}_{i}-\vec{y}_{i}}.\label{eq:metric-function-2}
\end{equation}

Now we can see the shortcomings of our theory: it relies too much on the naked-eye observation.
Given $\vec{S}_0$ and $\Psi$, only by observing the generated sequence $(\vec{S}_{t})_{\mathbb{N}}$ and its derived sequences can we learn about their convergence or disconvergence, their periodicity, \eqpt{}s being ``attractors'' or ``repellors'', \eqpt{} approximation accuracy, etc.
And yet we fail to, in the case of two-person game, provide a function if any {--} say $\chi(\vec{A},\vec{B},r_r,r_c,\vec{S}_0)$ {--} to conveniently determine for $[\Psi,\vec{S}_0]$.
Nor can we provide a function $\chi(\vec{A},\vec{B},r_r,r_c)$ to determine whether every $\vec{S}_0{\in}\mathbb{S}$ leads to a convergent sequence.

As a practical use of our theory, given a game and players' initial strategy $\vec{S}_0$ we need to find out the optimal parameters for function $\Psi$ in order to optimize \eqpt{} approximation accuracy.
That is, in the case of two-person game, given payoff bimatrix $(\vec{A},\vec{B})$ and initial $(\vec{s}_r,\vec{s}_c)_0$, we need to find out the optimal input $[(\vec{A},\vec{B}),(r_r, r_c),(\vec{s}_r,\vec{s}_c)_0]$ for Algorithm~\ref{alg:two-person} to minimize overall regret sum.
For that purpose, obviously we can try different adjustment rates $(r_r, r_c)$ in input as in FIG~\ref{fig:find-best-Psi-with-rate}.
And notice that bimatrix $(a_r\vec{A}{+}b_r,a_c\vec{B}{+}b_c)$, where $(a_r,a_c){\in}\mathbb{R}_{+}^2$ and $(b_r,b_c){\in}\mathbb{R}^2$, has the same \eqpt{}s as bimatrix $(\vec{A},\vec{B})$.
Thus we can try different scale factors $(a_r,a_c)$ in input as shown in FIG~\ref{fig:find-best-Psi-with-scale}.
With parameters $(r_r,r_c)$, $(a_r,a_c)$ and $(b_r,b_c)$ combined, there forms a search space for function $\Psi$, which could be overwhelming when it comes to many-person games.

\begin{figure}[h]
\centering
\includegraphics[width=\linewidth]{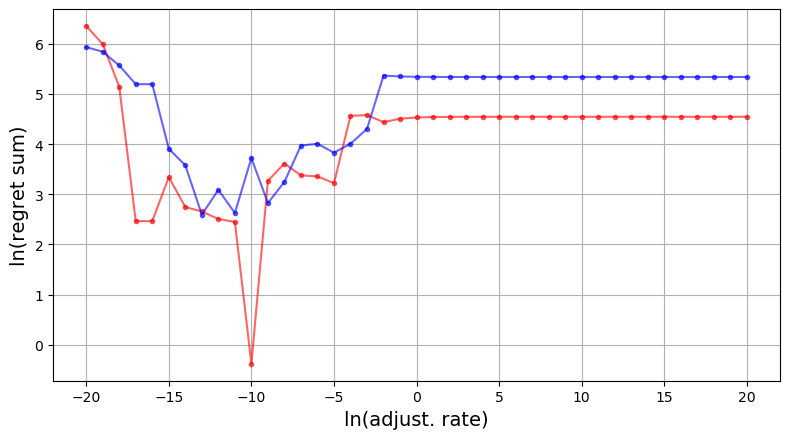}
\caption{\label{fig:find-best-Psi-with-rate}
    The approximate \eqpt{}'s regret sum $(R_r,R_c)^{*}$ of Algorithm~\ref{alg:two-person} for different adjustment rates in the case of game \game{3X3-1eq3sp}.
    For simplicity we consider $r_r{=}r_c$ for the input $(r_r,r_c)$.
}
\end{figure}

\begin{figure}[h]
\centering
\includegraphics[width=\linewidth]{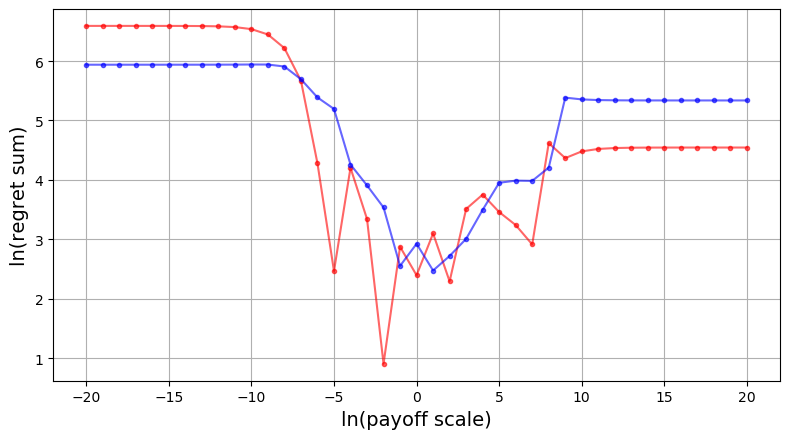}
\caption{\label{fig:find-best-Psi-with-scale}
    The approximate \eqpt{}'s regret sum $(R_r,R_c)^{*}$ of Algorithm~\ref{alg:two-person} for different bimatrix scale factors in the case of game \game{3X3-1eq3sp}.
    For simplicity we consider $a_r{=}a_c$ and $(b_r,b_c){=}(0,0)$ for the input $(a_r\vec{A}{+}b_r,a_c\vec{B}{+}b_c)$.
    Note that the output $(R_r,R_c)^{*}$ is ``scaled back'' with $(a_r^{-1},a_c^{-1})$ for proper comparison.
}
\end{figure}

\section{Variants of function $\lambda_{i,\vec{S}}$\label{section:lambda_variants}}

In this section, we will introduce two variants of $\lambda_{i,\vec{S}}$ function for the purpose of generalization.

By applying on each component of the vector $\lambda_{i,\vec{S}}(\vec{s}_i)$ a general function $\alpha_{ij}:\mathbb{R}_{{\geq}0}{\rightarrow}\mathbb{R}_{{\geq}0}$, we have the first variant:

\begin{align}
\nonumber
&\bar{\lambda}_{i,\vec{S}}(\vec{s}_i)= \\
&\Big(\alpha_{i1}{\circ}\varphi_{i1,\vec{S}}(\vec{s}_i),\ldots,\alpha_{ij}{\circ}\varphi_{ij,\vec{S}}(\vec{s}_i),\ldots,\alpha_{ig_i}{\circ}\varphi_{ig_i,\vec{S}}(\vec{s}_i)\Big).
\label{eq:lambda-with-alpha}
\end{align}

Accordingly, we can write Eq. (\ref{eq:increase-payoff-1}) to
\begin{align}
\nonumber
&\inprod{\bar{\lambda}_i(\vec{s}_i)}{\vec{v}_i}-\inprod{\vec{s}}{\vec{v}_i}\abs{\bar{\lambda}_i(\vec{s}_i)} \\
&=\sum_{j}^{g_i}\alpha_{ij}{\circ}\varphi_{ij}(\vec{s}_i)\Big(\bar{p}_i(\pi_{ij};\vec{S})-p_i(\vec{s}_i)\Big).
\label{eq:increase-payoff-21}
\end{align}

Suppose that $\alpha_{ij}(0){=}0$ for any $i$ and $j$, then as with Eq. (\ref{eq:increase-payoff-1}), in Eq. (\ref{eq:increase-payoff-21}) the summation terms with $\bar{p}_i(\pi_{ij};\vec{S}){-}p_i(\vec{s}_i){\le}0$ must be zeros or otherwise $\alpha_{ij}{\circ}\varphi_{ij}(\vec{s}_i)\varphi_{ij}(\vec{s}_i)$, such that
\begin{align}
\nonumber
&\sum_{j}^{g_i}\alpha_{ij}{\circ}\varphi_{ij}(\vec{s}_i)\Big(\bar{p}_i(\pi_{ij};\vec{S})-p_i(\vec{s}_i)\Big) \\
&=\sum_{j}^{g_i}\alpha_{ij}{\circ}\varphi_{ij}(\vec{s}_i)\varphi_{ij}(\vec{s}_i)=\inprod{\bar{\lambda}_i(\vec{s}_i)}{\lambda_i(\vec{s}_i)}.
\label{eq:increase-payoff-22}
\end{align}

With Eq. (\ref{eq:increase-payoff-21}) and Eq. (\ref{eq:increase-payoff-22}) we can write Eq. (\ref{eq:increase-payoff-2}) to

\begin{equation}
p_{i}(\vec{s}'_i)=p_{i}(\vec{s}_i)+\frac{\inprod{\bar{\lambda}_i(\vec{s}_i)}{\lambda_i(\vec{s}_i)}}{r^{-1}_i+\abs{\bar{\lambda}_i(\vec{s}_i)}}\ge{}p_{i}(\vec{s}_i).\label{eq:increase-payoff-3}
\end{equation}

It is important to note that, $\alpha_{ij}(0){=}0$ also ensures, as discussed in Section~\ref{section:fpi_eqpt}, ``less profitable than average'' pure strategies being less used.
Now, the new $\bar{\lambda}_{i,\vec{S}}$ function not only increases payoff and thus decreases regret sum for each player's unilateral strategy adjustment, but also exposes a set of $\alpha_{ij}$ functions to provide more control on the process of fixed point iteration.
For example, as Eq. (\ref{eq:lambda-with-alpha}) implies, the set of $\alpha_{ij}$ functions directly controls the target of player $i$'s strategy adjustment by manipulating the components of regret vector.
In fact, Eq. (\ref{eq:lambda-with-alpha}) can be seen as a generalization of the decision-making on strategy target, while Eq. (\ref{eq:lambda}) is a special case with each $\alpha_{ij}$ being identity function.
In addition to functions $\alpha_{ij}$, we notice that Eq. (\ref{eq:increase-payoff-3}) still holds true if the constant $r_i$ is replaced with a general function $r_i:\mathbb{S}{\rightarrow}\mathbb{R}_{+}$, which allows more control on the adjustment rate for the process of fixed point iteration.
We believe that functions $\alpha_{ij}$ and $r_i$ could broaden the scope for the interpretation of our theory in the real-life cases.
With them, Eq. (\ref{eq:psi}) can be generalized to

\begin{equation}
\psi_{i,\vec{S}}(\vec{s}_i)=\frac{\vec{s}_i+r_i(\vec{S})\bar{\lambda}_{i,\vec{S}}(\vec{s}_i)}{1+r_i(\vec{S})\abs{\bar{\lambda}_{i,\vec{S}}(\vec{s}_i)}}.\label{eq:generalized_psi}
\end{equation}

Next we introduce the second variant of $\lambda_{i,\vec{S}}$ to propose another form of fixed point iteration beyond the non-cooperative game.
In no relevance to the concepts of game and regret, here we consider a system of $n$ agents whose states, alike strategies, are defined to be points in simplex.
Most distinctively, player $i$'s $\lambda_{i,\vec{S}}:\mathbb{I}_i{\rightarrow}\mathbb{R}_{{\geq}0}^{g_i}$ of Eq. (\ref{eq:lambda}) is replaced with a general continuous $\lambda_{i,\vec{S}}:\mathbb{I}_i{\rightarrow}\mathbb{I}_i$ which maps agent $i$'s current state into a target state.
For every agent $i$, because simplex $\mathbb{I}_i$ is a compact and convex set, according to fixed-point theorem~\cite{fixedpoint} given any $n$-tuple $\vec{S}$ there must exist an $\vec{s}_i^+{\in}\mathbb{I}_i$ such that $\vec{s}_i^+{=}\lambda_{i,\vec{S}}(\vec{s}_i^+)$.
Meanwhile, $\psi_{i,\vec{S}}(\vec{s}_i)$, now representing agent $i$'s state updating over $\vec{S}$, has the same geometrical interpretation as in FIG~\ref{fig:angle_close_up} except for vector $\lambda_{i,\vec{S}}(\vec{s}_i)$ always being on simplex $\mathbb{I}_i$.
Then with $\abs{\lambda_{i,\vec{S}}(\vec{s}_i)}{=}1$ we can write Eq. (\ref{eq:psi}) to

\begin{equation}
\psi_{i,\vec{S}}(\vec{s}_i)=\frac{1}{1+r_i}\vec{s}_i+\frac{r_i}{1+r_i}\lambda_{i,\vec{S}}(\vec{s}_i).
\label{eq:psi-variant}
\end{equation}

Here $r_i{>}0$.
Therefore, vector $\vec{s}'_i{=}\psi_i(\vec{s}_i)$ is a convex combination of vectors $\vec{s}_i$ and $\lambda_i(\vec{s}_i)$;
$\Psi(\vec{S}^+){=}\vec{S}^+$ if and only if each item $\vec{s}_i^+$ in $\vec{S}^+$ satisfies $\lambda_{i,\vec{S}^+}(\vec{s}^+_i){=}\vec{s}^+_i$.
And we have

\begin{align}
\nonumber
&\norm{\vec{s}'_i-\lambda_i(\vec{s}_i)}=\Big\lVert\frac{1}{1+r_i}\vec{s}_i-\frac{1}{1+r_i}\lambda_i(\vec{s}_i)\Big\rVert\\
&=\frac{1}{1+r_i}\norm{\vec{s}_i-\lambda_i(\vec{s}_i)}<\norm{\vec{s}_i-\lambda_i(\vec{s}_i)}.\label{eq:reduce-L2-distance}
\end{align}

This inequality states that, each agent $i$'s state will get closer, in terms of $L2$-norm distance, to its target state $\lambda_i(\vec{s}_i)$ if it unilaterally updates its state by Eq. (\ref{eq:psi-variant}).
Given proper function $\lambda_i$ and update rate $r_i$, $\norm{\vec{s}_i-\lambda_i(\vec{s}_i)}$ could be decreasing, even monotonically, along  the sequence $(\vec{S}_{t})_{\mathbb{N}}$ of Eq. (\ref{eq:sequence}).
Now as with the treatment in Section~\ref{section:fpi_eqpt} let us ideally assume that, $\lim_{\infty}\norm{\vec{s}_i{-}\lambda_i(\vec{s}_i)}^{(t)}{=}0$ or equivalently $\lim_{\infty}\big(\vec{s}_i{-}\lambda_i(\vec{s}_i)\big)^{(t)}{=}\vec{0}$ for all agents even if they simultaneously update their states by Eq. (\ref{eq:psi-variant}).
In that case, $(\vec{S}_{t})_{\mathbb{N}}$ is a process of fixed point iteration such that $\Psi(\vec{S}^+){=}\vec{S}^+$ with $\vec{S}^+{=}\lim_{\infty}\vec{S}_{t}$.
In our study on the firing equilibrium of neural network~\cite{Lan2014,Lan2018}, for this form of fixed point iteration we proposed two special cases: one considered $n{=}1$ agent, and the other one considered $n{>}1$ agents each having $\vec{s}_i{\in}\mathbb{I}_i{\subset}\mathbb{R}^2$.

We can compare this new fixed point iteration (FPI\textsubscript{2} to refer to) with the one in Section~\ref{section:fpi_eqpt} (FPI\textsubscript{1}).
Either FPI\textsubscript{1} or FPI\textsubscript{2} seeks out a fixed point $\vec{S}^+$ of function $\Psi$ such that $\Psi(\vec{S}^+){=}\vec{S}^+$.
In FPI\textsubscript{1} such $\vec{S}^+$ is an equilibrium point of non-cooperative game because $\abs{\lambda_{i,\vec{S}^+}(\vec{s}^+_i)}{=}0$ for all players, while in FPI\textsubscript{2} such $\vec{S}^+$ is a ``stationary'' state of system because $\lambda_{i,\vec{S}^+}(\vec{s}^+_i){=}\vec{s}^+_i$ for all agents.
In FPI\textsubscript{1} each player adjusts its strategy towards vector $\lambda_i(\vec{s}_i)$ to diminish it, while in FPI\textsubscript{2} each agent updates its strategy towards $\lambda_i(\vec{s}_i)$ to become it.
We believe that, as with FPI\textsubscript{1}, the metric functions of Eq. (\ref{eq:metric-function-1}) and Eq. (\ref{eq:metric-function-2}) could be applied to FPI\textsubscript{2} for useful results.

\section{Conclusions\label{section:conclusions}}
In this paper, to approximate a Nash equilibrium of non-cooperative game, we propose and test a regret matching with geometrical flavor.
The results show that, if each and every player continuously and ``smoothly'' suppresses the use of immediately ``unprofitable'' pure strategies, the game will evolve towards an equilibrium point.
Our approach has its merit in its simplicity and naturalness, suggesting that in the really world the tendency towards Nash equilibrium could be more pervasive than expected.
In natural selection the ``unprofitable'' pure strategies are immediately suppressed because the players using them have less offspring to pervade them.
In markets the players tend to use less in their mixed ``portfolios'' the pure strategies they deem immediately unprofitable.
And we can imagine that in many cases the tendency could be too strong to resist, so much so that its prevention could perhaps be a more suitable topic, since most likely the equilibrium point is not optimal against the non-equilibrium ones in terms of overall payoff.

There are some other points worth mention.
Our approximation of equilibrium point, which all boils down to players' reduction of regret sum, actually makes no assumption on the existence of equilibrium point; the players need no notion of equilibrium point.
The players adjust their strategies nearly independently with minimal information exchange among them, which means that the approximation of equilibrium point is effectively a distributed computation carried out by the nature.
The periodicity we observe in those sequences of Section~\ref{section:2p_game} and Section~\ref{section:accuracy} is an unexpected and yet truly interesting pattern worth further investigations.


\bibliographystyle{unsrt}
\bibliography{refs}

\end{document}